\documentclass[conference]{IEEEtran}
\IEEEoverridecommandlockouts
\usepackage[utf8]{inputenc}
\usepackage{cite}
\usepackage{amsmath,amssymb,amsfonts}
\usepackage{amsthm}
\usepackage{algorithm}
\usepackage{algorithmic}
\usepackage{bm}
\usepackage{graphicx}
\usepackage{textcomp}
\usepackage{xcolor}
\usepackage{physics}
\usepackage{booktabs}
\usepackage{hyperref}
\usepackage{subfig}
\hypersetup{colorlinks=true, linkcolor=black, citecolor=black, urlcolor=black}
\def\BibTeX{{\rm B\kern-.05em{\sc i\kern-.025em b}\kern-.08em
		T\kern-.1667em\lower.7ex\hbox{E}\kern-.125emX}}
\theoremstyle{IEEEproof}
\newtheorem{theorem}{Theorem}
\newtheorem{lemma}{Lemma}

\theoremstyle{remark}
\newtheorem{remark}{Remark}

\theoremstyle{definition} 
\newtheorem{definition}{Definition}

\newcommand{\Oc}{\Omega_c}
\newcommand{\Op}{\Omega_p}
\newcommand{\Os}{\Omega_s}
\newcommand{\Dc}{\Delta_c}
\newcommand{\Dp}{\Delta_p}
\newcommand{\Ds}{\Delta_s}

\begin{document}

\linespread{0.89}

\title{Unified Analytical Model for Atomic Receivers Under Typical Quantum Interference Paths\\
}

\author{Yiyue Xiang\textsuperscript{1}, 
	Neng Ye\textsuperscript{1},
	Qihao Peng\textsuperscript{2},
	Pei Xiao\textsuperscript{2},
	Jianping An\textsuperscript{1}\\
	\IEEEauthorblockA{\textsuperscript{1} School of Cyberspace Science and Technology, Beijing Institute of Technology, Beijing, China}
	\IEEEauthorblockA{\textsuperscript{2} 5G and 6G Innovation Centre, Institute for Communication Systems of the University of Surrey, Guildford, GU2 7XH, UK.}
}

\maketitle

\begin{abstract}
	Atomic receivers, which leverage the quantum interference termed electromagnetically induced transparency (EIT) for radio-frequency (RF) to optical signal transduction, offer a revolutionary paradigm for next-generation wireless communications. 
	However, current information-theoretic characterizations are predominantly restricted to the $\Xi$-type of EIT path and rely heavily on the weak-probe approximation, which fails to predict the behavior of the atomic receivers under high signal-to-noise ratio regimes. In this paper, we establish a unified analytical model for atomic receivers, and apply this model to three typical quantum interference paths, i.e., $V$-type, $\Lambda$-type, and $\Xi$-type configurations. 
	To provide a universal characterization, we propose the quantum coherence transfer coefficient (QCTC) to model the equivalent channel response induced by atomic receivers, using a steady-state perturbation framework built on the three-level EIT solution. The closed-form expressions of equivalent channel gains are then derived for three paths.
	Our results provide an analytical foundation for future capacity analysis and waveform optimization in atomic radio communication.
\end{abstract}


\section{Introduction}
Modeling the equivalent channel response induced by realistic radio receivers, e.g., power amplifiers nonlinearity \cite{abou2025channel,9666889} and local oscillators  phase noise \cite{ghozlan2017models, piemontese2022new, 6408166}, is the cornerstone of modern information theory.
It serves as an indispensable prerequisite for channel capacity bound derivation, optimal coding scheme design, and end-to-end performance optimization \cite{xie2004network,telatar2002capacity,cheng2023degree}.
For any communication system, a rigorous and accurate channel model forms the basis for characterizing its theoretical performance limits and guiding practical implementation.

Atomic receivers, as a revolutionary technology for future wireless communications, leverage quantum interference effect in atomic ensemble to realize radio frequency (RF) signal reception \cite{anderson2020atomic}. Unlike conventional receivers, it breaks through the Chu limit on metal antennas and the thermal noise floor of electronic devices, offering ultra-wideband spectrum compatibility and exceptional sensitivity \cite{schlossberger2024rydberg,meyer2020assessment}. However, the lack of a rigorous information-theoretic characterization of their equivalent model has become a critical bottleneck restricting the engineering application of atomic receivers.

The core operating principle of atomic receivers is the quantum interference effect known as electromagnetically induced transparency (EIT) \cite{harris1997electromagnetically}. In an EIT-based receiver, the weak RF signal modulates the quantum coherence of the atomic ensemble, which in turn modulates the transmission of an optical probe laser, enabling the RF-to-optical signal conversion for subsequent detection. Physically, EIT can be realized via three interference paths: V-type, $\Lambda$-type, and $\Xi$-type \cite{khan2016role}. These configurations exhibit distinct performance trade-offs across key metrics such as RF sensitivity, instantaneous bandwidth, and anti-interference capability  \cite{sen2015comparison}. However, existing information-theoretic studies on atomic receivers almost exclusively focus on the $\Xi$-type configuration \cite{zhu2025general,yuan2025electromagnetic}, leaving the characterization of $V$-type and $\Lambda$-type configurations largely unexplored in the information theory community.

From a theoretical perspective, the steady-state response of an EIT-based atomic receiver is governed by the Lindblad master equation of a four-level quantum system, which essentially requires solving a $16 \times 16$ system of coupled linear differential equations \cite{chen2025harnessing}. Existing analytical solutions for the $\Xi$-type configuration rely on the overly strict weak-probe assumption ($\Op \ll \gamma, \Oc$), which forces the probe laser Rabi frequency to be far smaller than the atomic decay rate and coupling laser Rabi frequency \cite{jing2020atomic}. However, in practical experimental implementations, the probe laser intensity must be set to a level comparable to the atomic decay rate to achieve an acceptable signal-to-noise ratio. This mismatch leads to significant systematic discrepancies between theoretical predictions of existing models and actual experimental measurements.

In this paper, we address this theoretical gap by developing a unified perturbation model for EIT-based atomic receivers. Using the exact three-level EIT steady-state solution as the zero-order background and treating only the weak RF signal as perturbation, our model eliminates the restrictive weak-probe assumption of prior works \cite{chen2025harnessing,bussey2021rydberg,liao2020microwave}. We propose the quantum coherence transfer coefficient (QCTC) to characterize the linearized equivalent channel gain, bridging quantum dynamics and information-theoretic channel modeling. We also derive exact closed-form QCTC for  $V$-type, $\Lambda$-type, and $\Xi$-type EIT configurations, laying the theoretical foundation for the implementation atomic receiver implementation.

This paper is organized as follows. Section II introduces the general principle of atomic receivers and three quantum interference paths of EIT. Section III presents  main theoretical results, including the unified analytical framework and closed-form equivalent models for each EIT path. Section IV provides the numerical results, Section V concludes this paper.


\section{System Model and Typical Quantum Interference Paths}
In this section, we first introduce the general principle of atomic receivers and the equivalent channel formulation, then present the system models for the three EIT configurations.

\subsection{General Principle and Channel Formulation}

Atomic receivers exploit EIT interference effect in a four-level atomic ensemble to transduce the incident RF signals into measurable optical signals.  The input baseband signal $X(t)$ linearly modulates the RF electric field, defining the RF Rabi frequency $\Omega_s(t) \propto X(t)$. The atomic receiver system is also driven by two strong optical fields, i.e., a probe laser with Rabi frequency $\Omega_p$ and a coupling laser with Rabi frequency $\Omega_c$.

The quantum state evolution of the atomic ensemble is governed by the Lindblad master equation
\begin{equation}
	\dot{\bm{\rho}} = -\frac{i}{\hbar} [\bm{H}, \bm{\rho}] + \mathcal{L}[\bm{\rho}].
	\label{eq:lindblad_master}
\end{equation}
Here, $\bm{\rho}$ is the $4 \times 4$ Hermitian density matrix, whose diagonal elements $\rho_{ii} = \langle i | \bm{\rho} | i \rangle$ denote the population probability of level $|i\rangle$, and off-diagonal elements $\rho_{ij} = \langle i | \bm{\rho} | j \rangle, i \neq j$ denote the quantum coherence between levels $|i\rangle$ and $|j\rangle$. The Hamiltonian $\bm{H}$ comprises the unperturbed atomic energy structure and the dipole-interaction term with electromagnetic fields. The Lindblad dissipator $\mathcal{L}[\bm{\rho}]$ accounts for spontaneous emission and decoherence of atomic energy level.

For steady-state analysis, we set $\dot{\bm{\rho}} = 0$. The macroscopic output observable is the probe laser transmission $\Delta T_p$, which is proportional to the imaginary part of the probe transition coherence $\rho_{\text{probe}}$, i.e., $\Delta T_p \propto \text{Im}(\Delta\rho_{\text{probe}})$, and is subsequently converted to a photocurrent $Y(t)$ by a photodetector.

While the full mapping $X \to Y$ is intrinsically nonlinear due to the atomic quantum dynamics, it can be linearized under the reasonable weak RF signal assumption, i.e.,
\begin{equation}
	\Omega_s \ll \Omega_c, \Omega_p, \gamma_3, \gamma_4,
	\label{eq:perturbation_hierarchy}
\end{equation}
where $\gamma_i$ is the spontaneous decay rate of state $|i\rangle$. Under this condition, we can abstract the equivalent channel model as
\begin{equation}
	Y = \mathcal{H}_q(\Omega_c, \Omega_p, \Delta) \cdot X + Z,
	\label{eq:channel_model}
\end{equation}
where $X$ and $Y$ are the normalized complex input and output, respectively. $Z$ denotes the equivalent additive noise bounded by photon shot noise and atomic spontaneous emission. The complex-valued coefficient $\mathcal{H}_q$, defined as the \textit{Quantum Coherence Transfer Coefficient (QCTC)}, encapsulates the linearized transduction gain. Crucially, $\mathcal{H}_q$ is a highly nonlinear function of the strong optical parameters $\Omega_c, \Omega_p$, and detunings $\Delta = \{\Delta_c, \Delta_p, \Delta_s\}$ of probe laser, coupling laser and RF signal, respectively, defined as the frequency offset from the transition resonance frequency $\omega_i$ of the target level $|i\rangle$.

This paper aims to derive the closed-form expressions of the nonlinear QCTC for atomic receiver, and to provide a unified model for characterizing its information-theoretic properties.

\subsection{Typical Quantum Interference Paths}
\subsubsection{V-Type Configuration}
The V-type EIT configuration has the energy level structure as follows 
\begin{itemize}
	\item $|1\rangle$: Ground state, no spontaneous decay, $\gamma_1 = 0$.
	\item $|2\rangle$: Excited state, coupled to $|1\rangle$ via the coupling laser with Rabi frequency $\Oc$, detuning $\Dc = \omega_2 - \omega_1 - \omega_c$.
	\item $|3\rangle$: Excited state, coupled to $|1\rangle$ via the probe laser with Rabi frequency $\Op$, detuning $\Dp = \omega_3 - \omega_1 - \omega_p$, and to $|4\rangle$ via the RF signal field.
	\item $|4\rangle$: RF-coupled state, coupled to $|3\rangle$ via the RF signal with Rabi frequency $\Os$, detuning $\Ds = \omega_4 - \omega_3 - \omega_s$.
\end{itemize}

The Hamiltonian for the V-type configuration is given by
\begin{equation}
	\bm{H}^{V} \!=\! \begin{pmatrix}
		0 & \Oc & \Op & 0 \\
		\Oc & -2\Dc & 0 & 0 \\
		\Op & 0 & -2\Dp & \Os \\
		0 & 0 & \Os & -2(\Dp\!+\!\Ds)
	\end{pmatrix}.
	\label{eq:hamiltonian_vtype}
\end{equation}


\subsubsection{$\Lambda$-Type Configuration}
The $\Lambda$-type EIT configuration has a metastable ground state $|2\rangle$ with negligible decay rate $\gamma_2 \approx 0$, which exhibits the following energy level structure
\begin{itemize}
	\item $|1\rangle$: Ground state, coupled to $|3\rangle$ via the probe laser with detuning $\Dp = \omega_3 - \omega_1 - \omega_p$.
	\item $|2\rangle$: Metastable ground state, coupled to $|3\rangle$ via the coupling laser with detuning $\Dc = \omega_3 - \omega_2 - \omega_c$.
	\item $|3\rangle$: Excited state, coupled to both ground states and the RF-coupled state $|4\rangle$.
	\item $|4\rangle$: RF-coupled state, coupled to $|3\rangle$ via the RF signal field with detuning $\Ds = \omega_4 - \omega_3 - \omega_s$.
\end{itemize}

The Hamiltonian for the $\Lambda$-type configuration is given by
\begin{equation}
	\bm{H}^{\Lambda} \!=\! \begin{pmatrix}
		0 & 0 & \Op & 0 \\
		0 & -2(\Dp\!-\!\Dc) & \Oc & 0 \\
		\Op & \Oc & -2\Dp & \Os \\
		0 & 0 & \Os & -2(\Dp\!+\!\Ds)
	\end{pmatrix},
	\label{eq:hamiltonian_lambdatype}
\end{equation}

\subsubsection{$\Xi$-Type Configuration}
The $\Xi$-type EIT configuration is the most widely studied in communication-related literature, with a cascaded energy level structure as follows
\begin{itemize}
	\item $|1\rangle$: Ground state, coupled to $|2\rangle$ via the probe laser with detuning $\Dp = \omega_2 - \omega_1 - \omega_p$.
	\item $|2\rangle$: Intermediate excited state, coupled to $|1\rangle$ via the probe laser and to $|3\rangle$ via the coupling laser with detuning $\Dc = \omega_3 - \omega_2 - \omega_c$.
	\item $|3\rangle$: Rydberg intermediate state, coupled to $|2\rangle$ via the coupling laser and to $|4\rangle$ via the RF signal field.
	\item $|4\rangle$: Higher Rydberg state, coupled to $|3\rangle$ via the RF signal field with detuning $\Ds = \omega_4 - \omega_3 - \omega_s$.
\end{itemize}

The Hamiltonian for the $\Xi$-type configuration is given by
\begin{equation}
	\bm{H}^{\text{$\Xi$}} \!=\! \begin{pmatrix}
		0 & \Op & 0 & 0 \\
		\Op & -2\Dp & \Oc & 0 \\
		0 & \Oc & -2(\Dp\!+\!\Dc) & \Os \\
		0 & 0 & \Os & -2(\Dp\!+\!\Dc\!+\!\Ds)
	\end{pmatrix}.
	\label{eq:hamiltonian_laddertype}
\end{equation}

The above Hamiltonians characterize the atom-field interaction dynamics for the three typical EIT configurations, laying the physical foundation for the analytical modeling. 

\section{Unified Analytical Channel Model for Atomic Response}
In this section, we first establish a unified steady-state perturbation framework tailored for EIT-based atomic receivers, then derive the closed-form expressions of QCTC under three EIT configurations for their equivalent analytical modeling.

\subsection{Unified Steady-State Perturbation Framework}
We adopt a perturbation hierarchy consistent with practical atomic receivers system, where the Rabi frequency $\Os$ of RF signal is the weakest energy scale. The perturbation hierarchy is given by
\begin{equation}
	\Os \ll \Oc, \quad \Os \ll \Op, \quad \Os \ll \gamma_3, \gamma_4.
	\label{eq:perturbation_hierarchy}
\end{equation}
Crucially, this hierarchy imposes no constraints on the relative magnitudes of $\Op$ with respect to $\Oc$ and the atomic decay rates, eliminating the error brought by weak-probe assumption.

Then, we expand all density matrix elements in powers of the weak RF signal $\Os$, i.e.,
\begin{equation}
	\rho_{ij} = \rho_{ij}^{(0)} + \rho_{ij}^{(1)} + \rho_{ij}^{(2)} + \cdots, \quad \rho_{ij}^{(n)} \sim \order{\Os^n},
	\label{eq:perturbation_expansion}
\end{equation}
where zero-order solution $\order{\Os^0}$ is the exact steady-state solution of the three-level EIT system when the RF signal is turned off, i.e., $\Os = 0$. In this case, the $|4\rangle$ completely decouples from the system, and all density matrix elements involving $|4\rangle$ vanish.
First-order correction $\order{\Os^1}$ is the signal-induced correction to the coherences involving the level $|4\rangle$.
Second-order correction $\order{\Os^2}$ is the leading-order correction to the probe coherence $\rho_{31}$ or $\rho_{21}$, which determines the output signal of the atomic receiver.

We first prove a lemma to simplify the perturbation analysis.
\begin{lemma}\label{lemma:first_order_pop_zero}
	For the three typical EIT configurations under the perturbation hierarchy \eqref{eq:perturbation_hierarchy}, the first-order corrections to the population terms are identically zero, i.e.,
	\begin{equation}\label{eq:first_order_pop_zero}
		\rho_{ii}^{(1)} = 0, \quad \forall i = 1,2,3,4.
	\end{equation}
\end{lemma}
\begin{IEEEproof}
	See Appendix A.
\end{IEEEproof}

Lemma \ref{lemma:first_order_pop_zero} indicates a physical interpretation that the population redistribution in the atomic system requires absorption and emission events, which involve the RF signal field to second order. At first order, the weak RF signal can only create quantum coherences, i.e., the off-diagonal density matrix elements, without changing the level populations.

With this lemma, we can now present our unified equivalent channel model.
\begin{definition}
	Let $\gamma_{ij}$ denote the decoherence rate between level $|i\rangle$ and $|j\rangle$, and $\Delta_k$ denote the detuning of field $k$.

	\noindent \textit{1. Complex Detuning-Decay Parameters (Bare States):}
	\begin{itemize}
		\item $D_p = \gamma_{ij} - i\Delta_p$: Probe Transition Detuning.
		\vspace{2pt}
		\item $D_c = \gamma_{kl} - i\Delta_c$: Coupling Transition Detuning.
		\item $D_{cp} = \gamma_{\text{TP}} - i(\Delta_p \pm \Delta_c)$: Two-Photon Detuning.
	\end{itemize}
	
	\vspace{2pt}
	
	\noindent \textit{2. Dressed Detuning Parameters (Optical Fields):}
	\begin{itemize}
		\item $\widetilde{D}_p = D_p + \Omega_c^2/(4 D_{cp}^{(\ast)})$:
		Probe transition energy dressed by the coupling field, where the conjugate symbol $(\ast)$ is configuration-specific.
		\vspace{2pt}
		\item $\widetilde{D}_c = D_c + \Omega_p^2/(4 D_{cp}^{(\ast)})$:
		Coupling transition energy dressed by the probe field, where the conjugate symbol $(\ast)$ is configuration-specific.
		\vspace{2pt}
		\item $\mathcal{D} = \widetilde{D}_p \widetilde{D}_c - \Omega_p^2 \Omega_c^2/16 |D_{cp}|^2$: Characteristic determinant of the steady-state optical Bloch equations for the three-level EIT background.
	\end{itemize}

	\vspace{2pt}
	\noindent \textit{3. RF-Coupled Transition Parameters:}
	\begin{itemize}
		\item $\mathfrak{d}_k$: Detuning-decay parameters for transitions involving the RF-coupled level $|4\rangle$. 
		
		\item $\widetilde{\mathfrak{d}}$: The effective detuning of the RF-coupled transition, dressed by both the probe and coupling optical fields.
	\end{itemize}
\end{definition}

\begin{theorem}[Unified Analytical Model]
	Under the perturbation hierarchy \eqref{eq:perturbation_hierarchy}, the equivalent channel gain (QCTC) of a four-level EIT-based atomic receiver can be given by
	\begin{equation}
		\mathcal{H}_q \triangleq \frac{\Delta \rho_{\text{probe}}}{\Omega_s^2} = \frac{-i \cdot \widetilde{D}_{\text{probe}}}{2 \mathcal{D}} \cdot \mathcal{T}_{41},
		\label{eq:unified_gain}
	\end{equation}
	 where $\Delta \rho_{\text{probe}}\!=\!\Delta \rho_{31}$ for V/$\Lambda$-type and $\Delta \rho_{21}$ for $\Xi$-type.  $\widetilde{D}_{\text{probe}}\!=\!\widetilde{D}_c^*$ for V-type, $\widetilde{D}_c$ for $\Lambda$-type, and  $\widetilde{D}_p^*$ for $\Xi$-type. $\mathcal{T}_{41} \equiv \rho_{41}^{(1)}/\Omega_s$ is the signal-to-coherence transfer coefficient. 
	\label{theorem:unified_channel}
\end{theorem}
\begin{IEEEproof}
	See Appendix B.
\end{IEEEproof}

Theorem \ref{theorem:unified_channel} provides a unified framework for the equivalent channel modeling of three types of EIT-based atomic receivers. The equivalent channel gain $\mathcal{H}_q$ is determined by the zero-order three-level EIT background and the first-order signal-induced coherence $\rho_{41}^{(1)}$. In the following subsections, we derive the closed-form expressions of $\mathcal{H}_q$ for each of the three EIT configurations.

\subsection{QCTC of the Three EIT Configurations}
\subsubsection{V-Type Configuration}
For the $V$-type configuration, we first specify the configuration-specific parameters.
The complex detuning-decay parameters are  
\begin{equation}
	D_p \!=\! \gamma_{13} - i\Delta_p, D_c \!=\! \gamma_{12} - i\Delta_c, D_{cp}\!=\! \gamma_{23} - i(\Delta_p \!-\! \Delta_c).
	\label{eq:vtype_detuning_defs}
\end{equation}
The RF-coupled transition parameters are specified as
\begin{align}
	\mathfrak{d}^{\text{V}}_1 &= \gamma_{14} - i(\Delta_p+\Delta_s), \quad \mathfrak{d}^{\text{V}}_2 = \gamma_{24} - i(\Delta_p+\Delta_s-\Delta_c), \nonumber \\
	\mathfrak{d}^{\text{V}}_3 &= \gamma_{34} - i\Delta_s, \quad \widetilde{\mathfrak{d}}^{\text{V}} = \mathfrak{d}^V_1 + \frac{\Oc^2}{4\mathfrak{d}^{\text{V}}_2} + \frac{\Op^2}{4\mathfrak{d}^{\text{V}}_3}.
	\label{eq:vtype_rydberg_defs}
\end{align}

Then, we present the closed-form equivalent channel model.
\renewcommand{\thetheorem}{2.1}
\begin{theorem}[V-Type QCTC]\label{V-Type QCTC}
	For the V-type atomic receiver, the signal-to-coherence transfer coefficient is given by
	\begin{equation}
		\mathcal{T}_{41}^{\text{V}} = \frac{1}{\widetilde{\mathfrak{d}}^{\text{V}}} \left( -\frac{i}{2} \rho_{31}^{(0)} + \frac{\Oc}{4\mathfrak{d}^{\text{V}}_2} \rho_{32}^{(0)} + \frac{\Op}{4\mathfrak{d}^{\text{V}}_3} \rho_{33}^{(0)} \right)
		\label{eq:chi41_vtype}
	\end{equation}
	where $\rho_{31}^{(0)}$, $\rho_{32}^{(0)}$, and $\rho_{33}^{(0)}$ are the zero-order steady-state coherences and population of the three-level EIT system, given by \eqref{eq:rho31_0}, \eqref{eq:rho32_0}, and \eqref{eq:pop_sol} in the Appendix A. 
	
	The closed-form QCTC of the $V$-type is obtained by substituting $\mathcal{T}_{41}^{V}$ into \eqref{eq:unified_gain}. 
	\label{theorem:vtype_channel}
\end{theorem}
\begin{IEEEproof}
	See Appendix C.
\end{IEEEproof}

\begin{remark}
	The closed-form expression in Theorem \ref{theorem:vtype_channel} reveals three distinct physical pathways through which the RF signal modulates the probe coherence: (a) the direct pathway via the pre-existing probe coherence $\rho_{31}^{(0)}$, (b) the coupling-mediated pathway through the EIT two-photon coherence $\rho_{32}^{(0)}$, and (c) the probe-mediated pathway via the probe-induced population $\rho_{33}^{(0)}$. All three pathways are fully captured in our model.
\end{remark}

\subsubsection{$\Lambda$-Type Configuration}
For the $\Lambda$-type configuration, the complex detuning-decay parameters are given by
\begin{equation}
	D_p \!=\! \gamma_{13} - i\Delta_p,  D_c = \gamma_{23} - i\Delta_c, 
	D_{cp} \!=\! \gamma_{12} - i(\Delta_p \!-\! \Delta_c).
	\label{eq:lambdatype_detuning_defs}
\end{equation}
The RF-coupled transition parameters are specified as
\begin{align}
	&\mathfrak{d}^\Lambda_1 = \gamma_{14} - i(\Delta_p+\Delta_s), \quad \mathfrak{d}^\Lambda_2 = \gamma_{24} - i(\Delta_c+\Delta_s),\nonumber \\
	& \mathfrak{d}^\Lambda_3 = \gamma_{34} - i\Delta_s,
	\widetilde{\mathfrak{d}}^\Lambda=\mathfrak{d}^\Lambda_{1}+\frac{\Op^2}{4(\mathfrak{d}^\Lambda_{3}+\frac{\Oc^2}{4\mathfrak{d}^\Lambda_{2}})}.
	\label{eq:lambdatype_rydberg_defs}
\end{align}

Then, we present the closed-form equivalent channel model for the $\Lambda$-type configuration.
\renewcommand{\thetheorem}{2.2}
\begin{theorem}[$\Lambda$-Type QCTC]
	For the $\Lambda$-type atomic receiver, the signal-to-coherence transfer coefficient is given by
	\begin{equation}
		\mathcal{T}^{\Lambda}_{41}=\frac{1}{\widetilde{\mathfrak{d}^{\Lambda}}}\left(-\frac{i}{2}\rho_{31}^{(0)}\!+\!\frac{i\Oc\Op}{8\mathfrak{d}^{\Lambda}_{2}(\mathfrak{d}^\Lambda_{3}\!+\!\frac{\Oc^2}{4\mathfrak{d}^\Lambda_{2}})}\rho_{32}^{(0)}\!+\!\frac{\Op}{4(\mathfrak{d}^\Lambda_{3}\!+\!\frac{\Oc^2}{4\mathfrak{d}^\Lambda_{2}})}\rho_{33}^{(0)}\right)
	\end{equation}
	The zero-order steady-state quantities $\rho_{31}^{(0)}$, $\rho_{32}^{(0)}$ and $\rho_{33}^{(0)}$ are derived from the three-level $\Lambda$-type EIT system.
	
	The closed-form QCTC of the $\Lambda$-type is obtained by substituting $\mathcal{T}_{41}^{\Lambda}$ into \eqref{eq:unified_gain}. 
	\label{theorem:lambdatype_channel}
\end{theorem}
\begin{IEEEproof}
	Proof omitted, similar to Theorem 2.1.
\end{IEEEproof}

\begin{remark}
	The $\Lambda$-type retains all three signal transduction pathways but with a modified coupling-mediated pathway. It changes from a first-order term $\Omega_c \rho_{32}^{(0)}$ in V-type to a second-order cross term $\Omega_c \Omega_p \rho_{32}^{(0)}$ in $\Lambda$-type. The long coherence lifetime of the metastable ground state $|2\rangle$ enables an ultra-narrow EIT window, making the $\Lambda$-type suitable for high-sensitivity RF sensing applications.
\end{remark}

\subsubsection{$\Xi$-Type Configuration}
For the $\Xi$-type configuration, the complex detuning-decay parameters are
\begin{equation}
	D_p \!=\! \gamma_{12} - i\Delta_p, D_c = \gamma_{23} - i\Delta_c, 
	D_{cp} \!=\! \gamma_{13} - i(\Delta_p \!+\! \Delta_c).
	\label{eq:xitype_detuning_defs}
\end{equation}
The RF-coupled transition parameters are specified as
\begin{align}
	\mathfrak{d}^{\Xi}_1 &= \gamma_{14} - i(\Delta_p+\Delta_c+\Delta_s), \quad \mathfrak{d}^{\Xi}_2 = \gamma_{24} - i(\Delta_c+\Delta_s), \nonumber \\
	\mathfrak{d}^{\Xi}_3 &= \gamma_{34} - i\Delta_s, \quad 
	\widetilde{\mathfrak{d}}^{\Xi} = \mathfrak{d}^{\Xi}_1 + \frac{\Op^2}{4(\mathfrak{d}^{\Xi}_2+\frac{\Oc^2}{4\mathfrak{d}^{\Xi}_3})}.
	\label{eq:xitype_rydberg_defs}
\end{align}
Then, we present the closed-form equivalent channel model for the $\Xi$--type configuration.
\renewcommand{\thetheorem}{2.3}
\begin{theorem}[$\Xi$-Type QCTC]
	For the $\Xi$-type atomic receiver, the signal-to-coherence transfer coefficient is given by
	\begin{equation}
		\mathcal{T}_{41}^{\Xi} =\frac{1}{\widetilde{\mathfrak{d}}^{\Xi}}\left(-\frac{i}{2}\rho_{31}^{(0)}\!+\!\frac{\Op}{4(\mathfrak{d}^{\Xi}_2+\frac{\Oc^2}{4\mathfrak{d}^{\Xi}_3})}\rho_{32}^{(0)}\!+\!\frac{i\Oc\Op}{8\mathfrak{d}^{\Xi}_{3}(\mathfrak{d}^{\Xi}_2+\frac{\Oc^2}{4\mathfrak{d}^{\Xi}_3})}\rho_{33}^{(0)}\right)
		\label{eq:chi41_xitype}
	\end{equation}
	The zero-order steady-state quantities $\rho_{31}^{(0)}$, $\rho_{32}^{(0)}$ and $\rho_{33}^{(0)}$ are derived from the three-level $\Xi$-type EIT system.
	
	The closed-form QCTC of the $\Xi$-type is obtained by substituting $\mathcal{T}_{41}^{\Xi}$ into \eqref{eq:unified_gain}.
	\label{theorem:xitype_channel}
\end{theorem}
\begin{IEEEproof}
	Proof omitted, similar to Theorem 2.1.
\end{IEEEproof}

\begin{remark}
	The analytical model of $\Xi$-type is fully consistent with existing literature under the weak-probe approximation, i.e., $\Op \to 0$, the probe-induced coherences $\rho_{31}^{(0)}$ and $\rho_{32}^{(0)}$ vanish, and \eqref{eq:chi41_xitype} reduces exactly to the first-order signal-to-coherence transfer coefficient of the standard weak-probe $\Xi$-type atomic receiver model. The proposed framework retains all high-order terms induced by finite probe power, extending the model's validity to arbitrary probe laser intensities.
\end{remark}

\subsection{Comparative Analysis}

The $V$-type configuration has the most complex model due to three first-order direct transduction pathways, which endows it with the highest tunability via coupling and probe lasers. The $\Lambda$-type configuration features a more compact form with a modified second-order coupling-mediated pathway, and its ultra-narrow EIT window from metastable ground-state long coherence lifetime makes it ideal for high-sensitivity applications. The $\Xi$-type configuration has the most concise closed-form model with minimal cross-coupling terms in its cascade transition structure, and it is the most widely adopted configuration in existing atomic receiver experiments for wireless communications.

\section{Numerical Results}
We validate the proposed unified analytical model via the comparisons with numerical simulations. The Rabi frequency of coupling laser is set to $\Omega_c = 2\pi \times 20$ MHz, and the RF signal Rabi frequency $\Omega_s$ scanned across the weak-signal regime defined by the perturbation hierarchy in (\ref{eq:perturbation_hierarchy}). The analytical results are calculated directly from the closed-form QCTC expressions derived in Section III, while the numerical results are obtained from the exact steady-state solutions of the Lindblad master equation without any approximation.

\begin{figure*}[!ht]
	\centering
	\subfloat[$V$-Type @ $\Delta_p=0$ MHz ]{\includegraphics[width=.5\columnwidth]{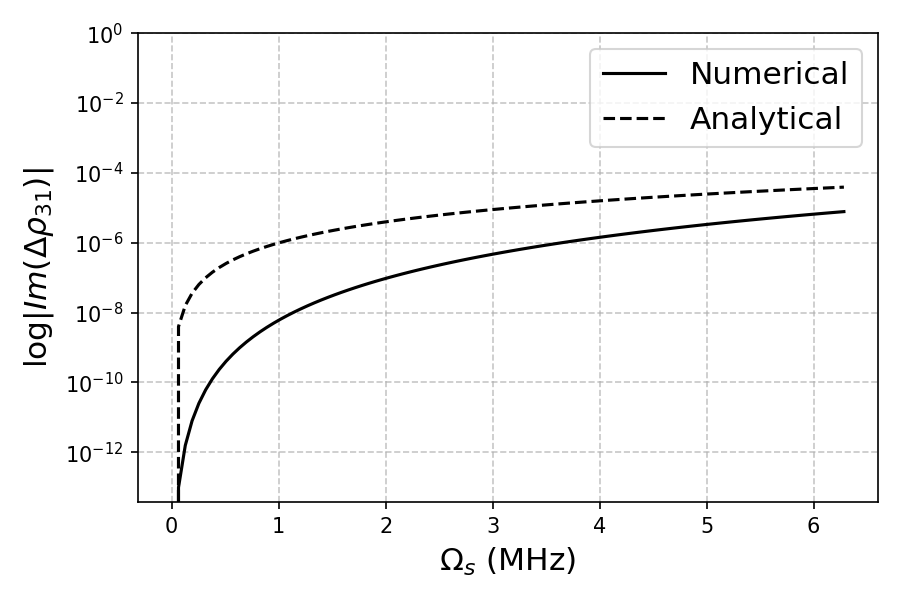} \label{subfig:v_0}} 
	\subfloat[$V$-Type @ $\Delta_p=2\pi $ MHz]{\includegraphics[width=.5\columnwidth]{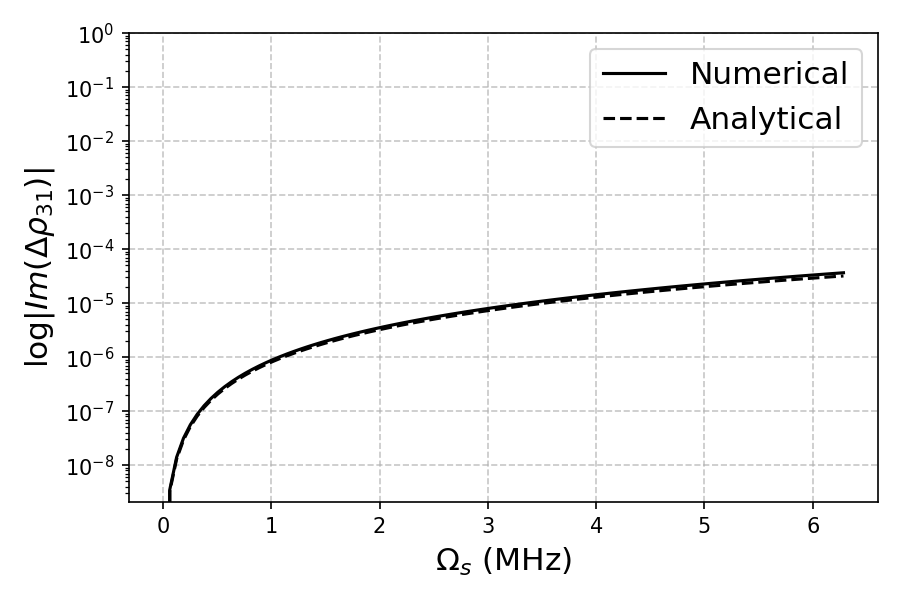} \label{subfig:v_2pi}} 
	\subfloat[$\Lambda$-Type @  $\Delta_p=0$ MHz]{\includegraphics[width=.5\columnwidth]{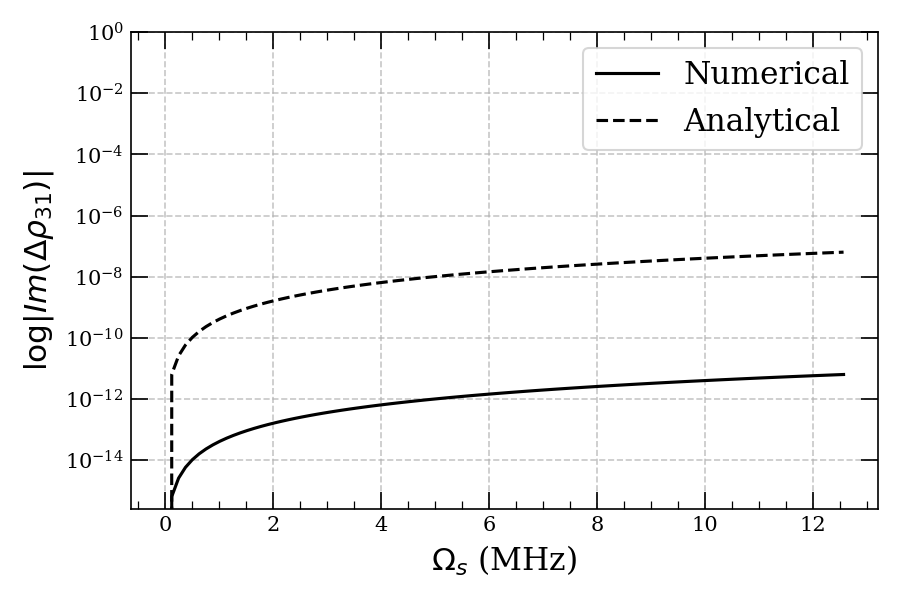} \label{subfig:lambda_0}} 
	\subfloat[$\Lambda$-Type @ $\Delta_p=2\pi $ MHz]{\includegraphics[width=.5\columnwidth]{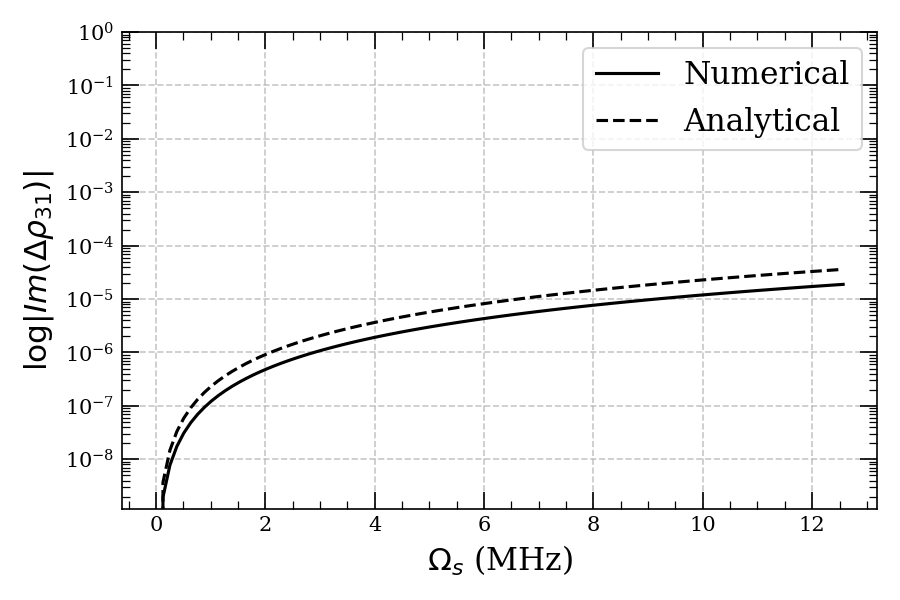} \label{subfig:lambda_2pi}} 
	\caption{Comparisons of the proposed QCTC analytical model and numerical simulation. These curves demonstrate the variations of probe laser coherence response with Rabi frequency of RF signal under different operation point $\Delta_p$.}
	\label{fig:model_validation}
\end{figure*}

Fig. \ref{fig:model_validation} compares the analytical and numerical results of the probe coherence response $|\text{Im}(\Delta\rho_{31})|$ for both V-type and $\Lambda$-type configurations. For the probe resonance case shown in Figs. \ref{subfig:v_0} and \ref{subfig:lambda_0}, our analytical model accurately captures the scaling trend of the probe coherence response with respect to $\Omega_s$. For the probe off-resonance case shown in Figs. \ref{subfig:v_2pi} and \ref{subfig:lambda_2pi}, the analytical results are in near-perfect agreement with the full numerical solutions across the entire weak-signal range. This directly validates that our proposed framework maintains high accuracy across different probe detuning conditions.
For the $\Xi$-type configuration, its closed-form QCTC expression is fully derived in Section III, and its steady-state EIT response characteristics have been validated in existing atomic receiver literature. We thus focus our numerical validation on the V-type and $\Lambda$-type configurations.

\section{Conclusion}
In this paper, we propose a unified analytical framework for EIT-based atomic receivers that eliminates the restrictive weak-probe assumption of prior models and aligns with practical experimental conditions. We derive, for the first time, closed-form equivalent channel gain models for all three canonical EIT configurations, i.e., $V$-type, $\Lambda$-type, and $\Xi$-type.
This work establishes an information-theoretic foundation for atomic communication system analysis and design. Future work will derive capacity bounds and analyze equivalent noise statistics to optimize atomic transceiver performance.

\appendices
\section{Proof of Lemma \ref{lemma:first_order_pop_zero}}
We first derive the zero-order steady-state solution of the three-level EIT system based on the Lindblad master equation.
Taking the V-type system as an example, eliminating the coherence $\rho_{32}^{(0)}$ and substituting it into the equations for $\rho_{21}^{(0)}$ and $\rho_{31}^{(0)}$, we can obtain the zero-order coherences as
\begin{align}
	\rho_{31}^{(0)} &= \frac{-i}{2\mathcal{D}^*} \left( \Op \widetilde{D}_c^* w_p - \frac{\Oc^2 \Op}{4 D_{cp}} w_c \right), \label{eq:rho31_0}\\
	\rho_{32}^{(0)} &= \frac{1}{D_{cp}}\left(\frac{-i\Op}{2}\rho_{12}^{(0)}+\frac{i\Oc}{2}\rho_{31}^{(0)}\right), \label{eq:rho32_0}
\end{align}
where $w_c = \rho_{11}^{(0)} - \rho_{22}^{(0)}$ and $w_p = \rho_{11}^{(0)} - \rho_{33}^{(0)}$ denote the population inversions.

The coupled linear differential equations are then simplified into a $2\times 2$ linear system for $(\rho_{22}^{(0)},\rho_{33}^{(0)})$ as
\begin{equation}
	(\gamma_2+2R_c-R_{cp})\rho_{22}^{(0)}+(R_c-2R_{cp})\rho_{33}^{(0)} = R_c-R_{cp},
\end{equation}
\begin{equation}
	(R_p-2R_{cp})\rho_{22}^{(0)}+(\gamma_3+2R_p-2R_{pc})\rho_{33}^{(0)} = R_p-R_{pc}.
\end{equation}
where $R_{c}\triangleq \Omega_c^2/2 {\rm Re}[\widetilde{D}_p^{*}/\mathcal{D}]$, $R_{p}\triangleq \Omega_p^2/2{\rm Re}[\widetilde{D}_c^{*}/\mathcal{D}]$, $R_{cp}\triangleq\Omega_c^2\Omega_p^2/8{\rm Re}[1/D_{cp}^{*}\mathcal{D}]$, $R_{pc}\triangleq\Omega_c^2\Omega_p^2/8{\rm Re}[1/D_{cp}\mathcal{D}^{*}]$.

Then, the zero-order populations are given by
\begin{equation}
	\begin{split}
			&\rho_{22}^{(0)} = \frac{A_{33} b_2 - A_{23} b_3}{A_{22} A_{33} - A_{23} A_{32}}, \rho_{33}^{(0)} = \frac{A_{22} b_3 - A_{32} b_2}{A_{22} A_{33} - A_{23} A_{32}}.
		\label{eq:pop_sol}
	\end{split}
\end{equation}
where $A_{22}\triangleq \gamma_2+2R_c-R_{cp}, A_{33}\triangleq \gamma_3+2R_p-2R_{pc},A_{23}\triangleq R_c-2R_{cp}, A_{32}\triangleq R_p-2R_{cp}$, $b_2\triangleq R_c-R_{cp}, b_3\triangleq R_p-R_{pc}$.

Considering the first-order extension of density matrix elements with respect to $\Os$, according to the diagonal population equations derived from the Lindblad master equation,
since the $\rho_{43}^{(0)}=0$,  we have $\rho_{44}^{(1)} = 0$ , and $\Op\,\text{Im}(\rho_{31}^{(1)}) = -\gamma_3\,\rho_{33}^{(1)}$.
The first-order off-diagonal equations for $\rho_{21}^{(1)}$, $\rho_{31}^{(1)}$, $\rho_{32}^{(1)}$ form a homogeneous linear system with the same coefficient matrix as the zero-order system, which is non-singular for all physically meaningful parameters. Therefore, the only solution is the trivial solution, i.e.,
\begin{equation}
	\rho_{21}^{(1)} = \rho_{31}^{(1)} = \rho_{32}^{(1)} = 0,
\end{equation}
which implies $\rho_{22}^{(1)} = \rho_{33}^{(1)} = 0$. From the trace constraint $\sum_{i=1}^4 \rho_{ii}^{(1)} = 0$, we get $\rho_{11}^{(1)} = 0$. This completes the proof.

\section{Proof of Theorem \ref{theorem:unified_channel}}
Owing to space constraints, we outline the universal derivation route, taking V/$\Lambda$-type (probe coherence $\rho_{31}$) as an example. The $\Xi$-type follows identically with $\rho_{21}$ and $\rho_{42}^{(1)}$. Detailed derivations are in the Appendix C.
By Lemma \ref{lemma:first_order_pop_zero}, all first-order populations vanish identically. Only Rydberg coherences $\rho_{41}^{(1)}, \rho_{42}^{(1)}, \rho_{43}^{(1)}$ are non-trivial, and solving their linear system gives the primary Rydberg coherence $\rho_{4k}^{(1)}$ in terms of zero-order EIT solutions.
The leading-order contribution is from $-\frac{i\Omega_s}{2}\rho_{4k}^{(1)}$. Solving via Cramer's rule gives
\begin{equation}
	\Delta \rho_{31} = \rho_{31}^{(2)} = \frac{-i \cdot \widetilde{D}_c^*}{2 \mathcal{D}^*} \cdot \mathcal{T}_{41} \cdot \Omega_s^2 + \mathcal{O}(\Omega_s^3),
\end{equation}
where $\mathcal{T}_{41} \equiv \rho_{41}^{(1)}/\Omega_s$. This generalizes to all configurations by substituting $\Delta \rho_{\text{probe}}$ and $\mathcal{T}_{4k}$, yielding \eqref{eq:unified_gain}. Higher-order terms are negligible under the perturbation hierarchy.

\section{Proof of Theorem \ref{V-Type QCTC}}
We first solve the first-order coherence equations for level $|4\rangle$ coherences, and 
the closed-form solution for the $\rho_{41}^{(1)}$ is
\begin{equation}
	\rho_{41}^{(1)} =\frac{1}{\widetilde{\mathfrak{d}}^V}\left(-\frac{i}{2}\Omega_s\rho_{31}^{(0)}+\frac{\Omega_c\Omega_s}{4\mathfrak{d}^V_{2}}\rho_{32}^{(0)}+\frac{\Omega_p\Omega_s}{4\mathfrak{d}^V_{3}}\rho_{33}^{(0)}\right),
\end{equation}
which defines the signal-to-coherence transfer coefficient as
\begin{equation}
	\mathcal{T}^V_{41} \equiv \frac{\rho_{41}^{(1)}}{\Omega_s} =\frac{1}{\widetilde{\mathfrak{d}}^V}\left(-\frac{i}{2}\rho_{31}^{(0)}+\frac{\Omega_c}{4\mathfrak{d}^V_{2}}\rho_{32}^{(0)}+\frac{\Omega_p}{4\mathfrak{d}^V_{3}}\rho_{33}^{(0)}\right).
\end{equation}

For second-order corrections $\mathcal{O}(\Omega_s^2)$, $\rho_{21}^{(2)}$ and $\rho_{31}^{(2)}$ retains the same homogeneous structure as the zero-order system, with additional source terms from first-order level $|4\rangle$ coherences. The full 2×2 system now reads
\begin{equation}
	\begin{pmatrix} 
		\widetilde{D}_{p}&-\eta^*\\-\eta^*&\widetilde{D}_{c} 
	\end{pmatrix} 
	\begin{pmatrix}  
		\rho_{21}^{(2)} \\ \rho_{31} ^{(2)} 
	\end{pmatrix} =  
	\begin{pmatrix} 
		- \frac{i}{2} \Omega_s \rho_{41}^{(1)}-\frac{i}{2}\Oc \omega_c^{(2)}\\ - \frac{i}{2} \Omega_s \rho_{42}^{(1)}-\frac{i}{2}\Op \omega_p^{(2)} 
	\end{pmatrix}.
\end{equation}

By Cramer’s rule, the signal-induced change in probe coherence is
\begin{align}
	\Delta \rho_{31} = -\frac{i\Os}{2\mathcal{D}}\left[\widetilde{D}_{c}\rho_{41}^{(1)}+\frac{\Oc\Op}{4D^*_{cp}}\rho_{42}^{(1)}\right] + \text{pop.} +\mathcal{O}(\Os^3)
\end{align}
where population corrections are negligible in the weak-signal limit. Substituting $\rho_{42}^{(1)}$ and $\mathcal{T}^V_{41}$ yields the unified QCTC form in Theorem \ref{theorem:unified_channel}. This completes the proof.

\bibliography{IEEEabrv,references} 

\end{document}